\documentclass[12pt,a4paper,final]{iopart}
\usepackage{graphicx}
\usepackage{dcolumn}
\usepackage{bm}
\usepackage{dsfont}
\usepackage{enumitem}
\usepackage{amsthm}
\usepackage{amsfonts}
\newtheorem*{theorem}{Theorem}

\usepackage{color} 

\begin{document}


\title{Classification of complex systems by their sample-space scaling exponents}

\author{Jan Korbel}%

\address{Section for Science of Complex Systems, CeMSIIS, Medical University of Vienna, Spitalgasse 23, 1090 Vienna, Austria}
\address{Complexity Science Hub Vienna, Josefst\"{a}dter Strasse 39, 1080 Vienna, Austria}
\ead{jan.korbel@meduniwien.ac.at}

\author{Rudolf Hanel}
\address{Section for Science of Complex Systems, CeMSIIS, Medical University of Vienna, Spitalgasse 23, 1090 Vienna, Austria}
\address{Complexity Science Hub Vienna, Josefst\"{a}dter Strasse 39, 1080 Vienna, Austria}
\ead{rudolf.hanel@meduniwien.ac.at}

\author{Stefan Thurner}
\address{Section for Science of Complex Systems, CeMSIIS, Medical University of Vienna, Spitalgasse 23, 1090 Vienna, Austria}
\address{Complexity Science Hub Vienna, Josefst\"{a}dter Strasse 39, 1080 Vienna, Austria}
\address{Santa Fe Institute, 1399 Hyde Park Road, Santa Fe, NM 87501, USA}
\address{IIASA, Schlossplatz 1, 2361 Laxenburg, Austria}
\ead{stefan.thurner@meduniwien.ac.at}

\begin{abstract}

The nature of statistics, statistical mechanics and consequently the thermodynamics
of stochastic systems is largely determined by how the number of states $W(N)$ depends on the size $N$ of the system.
Here we propose a scaling expansion of the phasespace volume $W(N)$ of a stochastic system.
The corresponding expansion coefficients (exponents) define the universality class the system belongs to.
Systems within the same universality class share the same statistics and thermodynamics.
For sub-exponentially growing systems such expansions have been shown to exist.
By using the scaling expansion this classification can be extended to {\em all} stochastic systems, including
correlated, constraint and super-exponential systems.
The extensive entropy of these systems can be easily expressed in terms of these scaling exponents.
Systems with super-exponential phasespace growth contain important systems, such as magnetic coins
that combine combinatorial and structural statistics.
We discuss other applications in the statistics of networks, aging, and cascading random walks.

\end{abstract}


\submitto{\NJP}                              
\maketitle


\section{Introduction}
Classical statistical physics typically deals with large systems composed of weakly interacting components,
which can be decomposed into (practically) independent sub-systems.
The phasespace volume $W$ or the number of states of such systems grows exponentially with system size $N$.
For example, the number of configurations in a spin system of $N$ independent spins is $W(N) = 2^N$.
For more complicated systems, however, where particles interact strongly, which are path-dependent, or whose
configurations become constrained, exponential phasespace growth no-longer occurs, and things become more interesting.
For example, in black holes the accessible number of states does not scale with the volume but with surface, which leads to
non-standard entropies and thermodynamics \cite{Beckenstein1974,Hawking1974,Thirring1970}.
A version of entropy that depends on the surface and the volume was recently suggested in \cite{biro}.

Other examples include systems with interactions on networks, path-dependent processes, co-evolving systems, and many
driven non-equilibrium systems. These systems are often non-ergodic and are referred to as {\em complex systems}.
For these systems, in general, the classical statistical description based on Boltzmann-Gibbs statistical mechanics fails to make correct predictions
with respect of the thermodynamic, the information theoretic, or the maximum entropy related aspects \cite{thurner_corominas_hanel17}.
Often the underlying statistics is then dominated by fat-tailed distributions, and power-laws in particular.
There have been considerable efforts to understand the origin of power-law statistics in complex systems.
Some progress was made for systems with sub-exponentially growing phasespace.
It was shown that systems whose phasespace grow as power laws, $W(N) \sim N^b$, are tightly related to so-called Tsallis statistics \cite{Sato2005}.

The tremendous variety and richness of complex systems has led to the question whether it is possible to classify them in terms of their statistical behavior.
Given such a classification, is it possible to arrive at a generalized concept of the statistical physics of complex systems, or do
we have to establish the statistical physics framework for every particular system independently?
For sub-exponentially growing systems such a classification was attempted by characterizing
stochastic systems in terms of two scaling exponents of their extensive entropy \cite{hanel-thurner11a}.
The first scaling exponent is recovered from the relation $\frac{S(\lambda W)}{S(W)} \sim \lambda ^c$,
which is valid if the first three Shannon-Khinchin axioms (see supplementary material) are valid (the fourth, the composition axiom, can be violated),
and if the entropy is of so-called trace form,
which means that it can be expressed as $S=\sum_i^W g(p_i)$, where $p_i$ is the probability for state $i$, and $g$ some function.
The second scaling exponent $d$ is obtained from a scaling relation that involves the re-scaling of the number of states  $W \to W^a$.
With these two scaling exponents $c$ and $d$ it becomes possible to classify sub-exponentially growing systems that fulfil the first three
Shannon-Khinchin axioms \cite{hanel-thurner11a}.
Further, the exponents $c$ and $d$ characterize the extensive entropy, $S_{c,d} \sim \sum \Gamma(1+d, c\log(p_i))$.
Practically all entropies that were suggested within the past three decades, are special cases of this $(c,d)$-entropy,
including Boltzmann-Gibbs-Shannon entropy ($c=1$, $d=1$),
Tsallis entropy ($d=0$),
Kaniadakis entropy ($c=1$, $d=1$) \cite{kaniadakis02},
Anteonodo-Plastino entropy  ($c=1$, $d>0$) \cite{Anteneodo1999},
and all others that  fulfil the first three Shannon-Khinchin axioms.
In \cite{hanel-thurner11b} it was then shown that the exponents $c$ and $d$ are tightly related with phasespace growth of the underlaying systems.
In fact, they can be derived from the knowledge of $W(N)$, ${1}/{(1-c)}= \lim_{N \to \infty} N W' /  W$,  and
$d= \lim_{N \to \infty}  \log  W  \left( W /  (NW') +c-1 \right)$.

For super-exponential systems such a classification is hitherto missing.
These systems include important examples of stochastic complex systems that
form new states as a result of the interactions of elements. These are systems that--besides
their combinatorial number of states (e.g. exponential)--form additional states that emerge as {\em structures}
from the components. The total number of states then grows super-exponentially with respect to system size,
e.g. the number of elements.
Stochastic systems with elements that can occupy several states (more than one) and that can form structures
with other elements, are generally super-exponential systems.
{It} was pointed out in \cite{Jensen16} that such systems might exhibit non-trivial thermodynamical properties.

An example for such systems are magnetic coins of the following kind.
Imagine a set of $N$ coins that come in two states, up and down. There are $2^N$
states. However, these coins are ``magnetic'', and any two of them can stick to each other, forming a new bond state (neither up nor down).
If there are $N=2$ coins, there are five states: the usual four states, uu, ud, du, dd, and a fifth state `bond'.
If there are $N=3$ coins, there are 14 states, the $2^3$ combinatorial states, and six states involving bond states:
state 9 is bond between coin 1 and 2, with the third  coin up, state 10 is the same bond state with the third coin down,
state 11 is a bond between 1 and 3 with the second con up, 12 the same bond with the second coin down, state 12 is a bond between
2 and 3,  with the first state up, and finally, state 14 is the bond between 2 and 3 with the first coin down.
It can be easily shown that the recursive formula for the number of states is,
$W_{}(N+1) = 2 W_{}(N) + N W_{}(N-1)$, which, for large $N$, grows as $W_{}(N) \sim N^{N/2} e^{2 \sqrt{N}}$, see \cite{Jensen16}.

In this paper we show that it is indeed possible to find a complete classification of complex stochastic systems, including the super-exponential case.
By expanding a generic phasespace volume $W_{}(N)$ in a Poincar\'{e} expansion, we will see that for any possibility
of phase space growth, there exists a sequence of unique expansion coefficients that are nothing but scaling exponents
that describe systems in their large size limit.
The {\em set} of scaling exponents gives us the full classification of complex systems in the sense that two systems belong to the same universality class,
if it is possible to rescale one into the other with exactly these exponents.
The framework presented here has been proposed in \cite{Hanel_thurner13}
and generalizes the classification approach of \cite{hanel-thurner11a,hanel-thurner11b}.
It includes the sub-exponential systems as a special case.
We show further that these exponents can be used straight forwardly to express--with a few additional requirements--the
corresponding extensive entropy, which is the basis for the thermodynamic properties of the system.
Finally, we see  in several examples that many systems are fully characterized by a very few exponents.
Technical details and auxiliary results are presented in the supplementary material. We reference to the supplementary material in the corresponding parts of the main text. However, readers may also go through the supplementary material before they continue reading.
We use the following notation for applying a function $f$ for $n$ times, $f^{(n)}(x) = \underbrace{f(\dots(f(x))\dots)}_{n \ times}$.

\section{Rescaling phasespace}

Suppose that phasespace volume depends on system size $N$ (e.g. number of elements) as $W(N)$.
We use the  Poincar\'{e} asymptotic expansion for the $l+1$ th logarithm of $W$,
\begin{equation}
	\log^{(l+1)} W(N) = \sum_{j=0}^{n} c_j \phi_{j}(N) + {\cal O} (\phi_n(N)) \quad ,
\label{Poin}
\end{equation}
where $\phi_{j}(N) = \log^{({j}+1)}(N)$ for $N \rightarrow \infty$.
A uniqueness theorem (see e.g. \cite{copson}) states that the asymptotic expansion exists and is uniquely determined for any $W(N)$
for which $\log^{(l+1)} W(N) = {\cal O} (\phi_0(N))$, see supplementary material.

To see how the exponents $c_j$ correspond to scaling exponents, let us define a sequence of re-scaling operations,
\begin{equation}
	r^{(n)}_\lambda(x) = \exp^{(n)}[\lambda \log^{(n)}(x)]\quad  .
\end{equation}
For example $r^{(0)}_\lambda(x) = \lambda x$, $r^{(1)}_\lambda(x) = x^\lambda$, etc. Obviously, $r^{(n)}_1(x) = x$.
The scaling operations obey the composition rule
\begin{equation}
	r^{(n)}_\lambda[r^{(n)}_{\lambda'}(x)] = r^{(n)}_{\lambda \lambda'}(x) \quad .
\label{eq:resc}
\end{equation}
We can now investigate the scaling behavior of the phasespace volume in the thermodynamic limit, $N \gg 1$.
The leading order of the scaling is given by the first rescaling $r_0$.
We show in the supplementary material that the rescaling of phasespace is asymptotically described by
\begin{equation}\label{eq:scaling}
	W(r^{(0)}_\lambda(N)) \sim r^{(l)}_{\lambda^{c^{(l)}_{0}}}(W(N)) \ \Rightarrow \ \frac{\log^{(l)} W(\lambda N)}{\log^{(l)} W(N)} \sim \lambda^{c^{(l)}_{0}} \quad ,
\end{equation}
where $c^{(l)}_{0} \in \mathds{R}$ is the leading exponent, and $l$ is determined from the condition that $c_0^{(l)}$ should be finite.
Thus, to leading order, the sample space grows as $W(N) \sim \exp^{(l)}\left(N^{c^{(l)}_{0}}\right)$.
We now identify the scaling laws for the sub-leading corrections through higher-order rescalings $W(r^{(k)}_\lambda(N))$.
We get (see supplementary material)
\begin{equation}\label{eq:sub}
	\frac{\log^{(l)} W(r^{(k)}_\lambda(N))}{\log^{(l)} W(N)} \prod_{j=0}^{k-1}
	\left(\frac{\log^{(j)}(r^{(k)}_\lambda(N))}{\log^{(j)}(N)}\right)^{- c^{(l)}_{j}} \sim \lambda^{c^{(l)}_{k}} \quad .
\end{equation}
Equivalently, one can express this relation as,
$W(r^{(k)}_\lambda(N)) \sim  r^{(l)}_{\sigma_k(N)}(W(N))$,
where $\sigma_k(N) = \prod_{j=0}^{k} \left(\frac{\log^{(j)}(r^{(k)}_\lambda(N))}{\log^{(j)}(N)}\right)^{c^{(l)}_{j}}$.
To extract $c^{(l)}_{j}$, take the derivative of Eq. (\ref{eq:scaling}) w.r.t. $\lambda$, set $\lambda=1$ and {consider} the limit $N \rightarrow \infty$.
For the leading scaling exponent we obtain
\begin{equation}
c^{(l)}_{0} = \lim_{N \rightarrow \infty} \frac{N W'(N)}{\prod_{i=0}^l \log^{(i)}W(N)} \quad . \label{depp}
\end{equation}
The scaling exponent corresponding to the $k$-th order is obtained in a similar way and reads,
\scriptsize
\begin{equation}\label{eq:clk}
	c^{(l)}_{k} =
	\lim_{N \rightarrow \infty}  \log^{(k)}(N) \left( \log^{(k-1)}(N)  \left( \dots \left( \log(N)  \left(\frac{N W'(N)}{\prod_{i=0}^l \log^{(i)}W(N)}-c^{(l)}_{0}\right)
	- c^{(l)}_{1}\right)\dots \right) -c^{(l)}_{(k-1)}\right)
\end{equation}
\normalsize
 This expression is not identically {equal} to zero, because the expression on the r.h.s. of Eq. (\ref{depp}) becomes $c^{(l)}_{0}$ only in the limit.
As a result, the phasespace volume grows as
\begin{equation}
	W(N) \sim \exp^{(l)}\left[\prod_{j=0}^{n} \left(\log^{(j)}(N)\right)^{c^{(l)}_{j}}\right] \quad,
\end{equation}
which is nothing but the Poincar\'{e} asymptotic expansion in Eq. (\ref{Poin}).
In the supplementary material we show that the formulas for $c_j$, given by the theory of asymptotic expansions,
correspond to the formulas for scaling exponents $c_j^{(l)}$ and therefore it is indeed possible to express {\em any} $W(N)$
in terms of an asymptotic expansion that is based on the sequence $\phi_n(N)$.
The expansion coefficients are scaling exponents determined by the rescaling of phasespace.
{Here $n$ denotes the minimal number of expansion terms. In the typical situations, only a few scaling exponents are non-zero. If all exponents are non-zero, we can truncate the expansion after a few terms and still preserve a high level of precision. In many realistic situations it is enough to consider $n=2$.}
The estimation of the leading order exponent can be tricky, 
because looking for the order $l$ incorporates calculation of several infinite limits.
Therefore, it is convenient to use an approach based on the corresponding extensive entropy.

\section{The extensive entropy}

The extensive entropy can be obtained by following an idea exposed in \cite{hanel-thurner11a,hanel-thurner11b}.
Let's assume a so-called {trace-form} entropy for some probability distribution $P = (p_1,\dots,p_W)$
\begin{equation}\label{eq:trace}
	S_g(p) = \sum_{i=1}^W g(p_i) \quad ,
\end{equation}
where $g$ is some function.
The aim is to find such a function $g$, for which the entropy functional $S_g$ is {\em extensive} for a given $W(N)$.
Assuming that no prior information about the system is given, we consider uniform probabilities $p_i = 1/W$.
The extensivity condition can be expressed by an equation for $g$, which is \cite{hanel-thurner11b}
\begin{equation}\label{eq:ext}
	S_{g}(W(N)) =  W(N)\, g(1/W(N)) \sim N \quad \mathrm{for} \ N \gg 1 \quad .
\end{equation}
{Alternatively, it is possible to define the extensive entropy as the solution of Euler's differential equation, see also \cite{biro},
\begin{equation}
N\, \frac{\mathrm{d} S(W(N))}{\mathrm{d}N} = S(W(N)).
\end{equation}}
The question now is, how the scaling exponents of $W(N)$ are related to scaling exponents of $S_g(W)$.
We begin with the first scaling operation $r^{(0)}$.
One can show that for $N \gg 1$, we have
\begin{equation}\label{eq:ent}
	S_g(r_\lambda^{(0)}(W)) \sim r^{(0)}_{\lambda^{d_0}}(S_g(W))
	\Rightarrow  \lambda \frac{g\left(\frac{1}{\lambda W(N)}\right)}{g\left(\frac{1}{W(N)}\right)} \sim \lambda^{d_0} \quad  .
\end{equation}
Thus, $g(x) \sim (1/x)^{d_0-1}$ for $x \rightarrow 0$.
Again, it is possible to determine the relation for the $n$ th scaling exponent
\begin{equation}\label{eq:entsub}
	\frac{g(1/r_\lambda^{(n)}(W))\, r_\lambda^{(n)}(W)}{g(1/W) \, W} \prod_{j=0}^{k-1} \left(\frac{\log^{(j)}
	(r_\lambda^{(n)}(W))}{\log^{(j)}(W)}\right)^{-d_j} \sim \lambda^{d_n} \quad ,
\end{equation}
or equivalently,
$S_g(r_\lambda^{(n)}(W)) \sim r_{\rho_n(W)}^{(0)}(S_g(W))$,
where $\rho_n(W) = \prod_{j=0}^{n} \left(\frac{\log^{(j)} (\lambda^{(k)}(W))}{\log^{(j)}(W)}\right)^{d_j}$.
We can extract the scaling exponents $d_n$ by the same procedure as for $c^{(l)}_{k}$ by taking the derivative w.r.t. $\lambda$,
setting $\lambda=1$ and performing the limit.
For the first exponent we get
\begin{equation}
	d_0 = \lim_{W \rightarrow \infty} \left(1 - \frac{g'(1/W)}{W g(1/W)}\right) \quad  .
\end{equation}
De L'Hospital's rule and applying the extensivity condition of Eq. (\ref{eq:ext}) gives $g'(W(N)) \sim N$, and
\begin{equation}
	d_0 = \lim_{N \rightarrow \infty} \frac{W(N)}{N W'(N)} \quad .
\end{equation}
We mentioned this result already above.
The $n$ th term can be found analogously to be
\scriptsize
\begin{equation}\label{eq:dn}
	d_n = \lim_{N \rightarrow \infty} \log^{(n)}(W) \left(\log^{(n-1)}(W) \left(\dots  \left( \log(W)  \left(\frac{W(N)}{N W'(N)}-d_0\right)
	-d_1 \right) \dots \right) -d_{n-1}\right)\quad  .
\end{equation}
\normalsize
We can now relate the scaling exponents $c^{(l)}_{k}$ and $d_n$ by comparing Eqs. (\ref{eq:clk}) and (\ref{eq:dn}).
For this we use a similar notation as for the exponents $c^{(l)}_k$ and assign $d^{(l)}_0 \equiv d_l$
to the first  non-zero exponent, $d_l \neq 0$.
All higher terms are denoted by $d^{(l)}_k = d_{l+k}$.
Using the fact that $N \sim (\log^{(l)}W)^{1/c^{(l)}_0}$, we finally obtain
\begin{eqnarray}
        \begin{array}{ll}
                    d^{(l)}_0 = \frac{1}{c^{(l)}_{0}}  \\[.3cm]
         d^{(l)}_k = -\frac{c^{(l)}_{k}}{c^{(l)}_{0}}, & k = 1,2,\dots \quad .\\
        \end{array}
\end{eqnarray}
The corresponding extensive entropy can now be characterized by the function $g(x)$,
which scales as
\begin{equation}\label{eq:g}
g^{{(l,n)}}(x) \sim x \prod_{j=0}^{n} \left(\log^{(j+l)} \frac{1}{x} \right)^{d^{(l)}_j} \quad \mathrm{for} \ x \rightarrow 0\, .
\end{equation}
the corresponding entropy scales as
\begin{equation}
	S^{{(l,n)}}_g(W) \sim \prod_{j=0}^{n} \left(\log^{(j+l)}W \right)^{d^{(l)}_j} \quad.
\label{scal}
\end{equation}
This equation is nothing but the asymptotic expansion of $\log S_g$ in terms of $\phi_{n+l}(N) = \log^{(n+l+1)}(N)$;
the coefficients are again the scaling exponents that correspond to the rescaling of the entropy.

Note that the entropy approach allows us to obtain additional restrictions for the scaling exponents if further information
about the system is available.
For example, many systems fulfil the first three of the four Shannon-Khinchin ({SK}) axioms, see supplementary material.
There we also show that it is possible to find a representation of the entropy that obeys the three axioms and the
scaling in Eq. (\ref{scal}).
In this case $g(x)$ can be expressed as
\begin{equation}\label{eq:entropy}
	g^{{(l,n)}}_{\left(d^{(l)}_0,\dots,d^{(l)}_n\right)}(x) = \int_0^{x} \prod_{i=0}^{n} \left[a_i + [1+\log]^{(i+l)}\left(\frac{1}{y}\right)\right]^{d^{(l)}_i} \mathrm{d} y \quad,
\end{equation}
where $a_i$ are constants. One possible choice for those is
\begin{equation}\label{eq:ai}
	a_i  =  \max\left\{ -1 - \frac{d^{(l)}_i}{(n-l)d^{(l)}_0},0\right\} \quad .
\end{equation}
The axioms impose restrictions on the range of scaling exponents.
({SK2}) requires that $d^{(l)}_0 > 0$; ({SK3}) requires that $d^{(0)}_0 \equiv d_0 < 1$.
The resulting entropy can be expressed by Eq. (\ref{eq:ent}).
One can trivially adjust the entropy minimal value, such that for the totally ordered state, $\mathcal{S}_g(1) = 0$.
This is obtained by rescaling
\begin{equation}
	\mathcal{S}_g(P) = r^{(-1)}_\lambda(S_g(P)) = \left(\sum_{i=1}^W g(p_i)\right) - g(1) \quad ,
\end{equation}
where $\lambda = \exp(g(1))$. Note that the form of the entropy in Eq (\ref{eq:entropy})
is equivalent to $(c,d)$-entropy for $c=1-d_0$ and $d=d_1$,
and $d_j = 0$ for all $j \geq 2$.

\section{Examples}

We conclude with several examples of systems that are characterized by different sets of scaling exponents.

{\em Exponential growth: the random walk}.
Imagine the ordinary random walk with two possibilities at any timestep--a step to the left, or to the right.
The number of possible configurations (i.e. possible paths) after $N$ steps is
\begin{equation}
	W_{}(N+1) = 2 W_{}(N) \quad,
\end{equation}
which means exponential phasespace growth, $W_{}(N) = 2^N$.
We obtain $l=1$, $c^{(1)}_0=1$ and $c^{(1)
}_j = 0$, for $j \geq 1$, and
for the exponents of the entropy $d_0 = 0$, $d_1\equiv d^{(1)}_0=1$ and $d_j = 0$, for $j \geq 2$.
This set of exponents belongs to the class of $(c,d)$-entropies described in \cite{hanel-thurner11a} for $c=1-d_0 = 1$, and $d=d_1=1$.
They correspond to the scaling exponents of the Shannon entropy:
from (\ref{eq:g}) we obtain that $g(x) \sim x \log x$ and from (\ref{scal}) we get $S(W) \sim \log W$, which is Boltzmann entropy.
It is not immediately apparent what the entropy of a random walk should be.
However, the random walk is equivalent to spin system of $N$ independent spins, the $2^N$ different paths correspond one-to-one to the
$2^N$ configurations in the spin model, where the role entropy of it is clear.
Obviously, for the random walk,  (SK 1-3) are applicable.

 {\em Sub-exponential growth: the aging random walk}.
In this variation of the random walk we impose correlations on the walk.
After the first random choice (left or right) the walker goes one step in that direction.
The second random choice is followed by two steps in the same direction, the next step is followed by three steps in the same direction, etc.
For $k$ independent choices, one has to make {$N = \sum_{i=1}^{k-1} i = 1/2 k(k-1)$ steps.
For this walk, we get that the number of possible paths is
\begin{equation}
	W_{}\left(N+k\right) = 2 W_{}\left(N\right) \quad ,
\end{equation}
which leads to $W(N) = 2^{N/k} \sim 2^{k/2}$. For $N \gg 1$, we have $k \approx \sqrt{N}$, and we obtain a stretched exponential (sub-exponential) asymptotic behavior, $W_{}(N) \sim 2^{\sqrt{N}}$.}
The order is again $l=1$ and the exponents are $c^{(1)}_0 = 1/2$ and $c^{(1)}_j = 0$, for $j \geq 1$.
In terms of the $d$ exponents we have $d_0=0$ and $d_1 \equiv d^{(1)}_0 =2$.
Therefore, the three SK axioms are applicable and the resulting extensive entropy belongs to the class of entropies
characterized by the Anteodo-Plastino entropy, since we have $g(x) \sim x (\log x)^2$ and $S(W) \sim (\log W)^2$.
This entropy is the special case of the $(c,d)$-entropy  for $c=1$ and $d=2$, see \cite{hanel-thurner11a}.

 {\em Super-exponential growth: magnetic coins}.
Consider $N$ coins with two states (up or down). These coins are magnetic, so that any two can stick to each other to create a pair
which is a third state obtained by interactions of elements (one possible configuration).
As mentioned before, in \cite{Jensen16} it is shown that the phasespace volume can be obtained recursively
\begin{equation}
	W_{}(N+1) = 2 W_{}(N) + N W_{}(N-1) \quad .
\end{equation}
For $N \gg 1$, we get $W_{}(N) \sim N^{N/2} e^{2 \sqrt{N}}$, which yields $l=1$, and the scaling exponents
$c^{(1)}_0 = 1$, $c^{(1)}_1 = 1$ and $c^{(1)}_j = 0$, for $j \geq 2$.
The scaling exponents of the entropy are $d_0=0$, $d_1 \equiv d^{(1)}_0=1$, and $d_2 \equiv d^{(1)}_1=-1$.
For the entropy this means, that $g(x) \sim x \log x/ \log \log x$ and $S(W) \sim \log W/\log \log W$.
This case is {\em not} contained in the class of $(c,d)$-entropies, because the third exponent,
corresponding to the doubly-logarithmic correction, is not zero. Actually we obtain $c=1$ and $d=1$,
which would naively indicate Shannon entropy. However, the correction makes the system clearly super-exponential.
The SK axioms are still applicable, the class of accessible entropy formulas is restricted by ({SK2}).
For example, for the representative entropy Eq. (\ref{eq:entropy}) we find that $a_0 \geq 0$ and $a_1 \geq 0$, see supplementary material.

{\em Super-exponential growth: random networks}.
Imagine a random network with $N$ nodes.
When a new node is added, there emerge $N$ new possible links, which gives us $2^N$ new possible configurations for each configuration of the network with $N$ links.
We obtain the recursive growth equation
\begin{equation}
	W_{}(N+1) = 2^N W_{}(N) \quad ,
\end{equation}
which leads to $W_{}(N) = 2^{{N}\choose{2}}$, as expected.
For this phasespace growth, we obtain $l=1$, $c_0^{(1)} = 2$ and  $c_j^{(1)} = 0$ for $j \geq 1$, and
$d_0=0$ and $d_1 \equiv d^{(1)}_0 =\frac{1}{2}$.
The corresponding entropy can be expressed by $g(x) \sim x (\log x)^{1/2}$,
and $S(W) \sim (\log W)^{1/2}$. The entropy corresponds to the class of compressed exponentials,
which are super-exponential, however, the entropy belongs to the class of $(c,d)$-entropies for $c=1$ and $d=1/2$.
Because all exponents are positive the entropy observes the SK axioms.

 {\em Super-exponential growth: the cascading random walk}.
 Consider a generalization of the random walk, where a walker can take a left or right step, but it can also
 split into two walkers, one of which then goes left, the other to the right.
Each walker can then go left, right, or split again (multiple walkers can occupy the same position).
The number of possible paths after $N$ steps is
\begin{equation}
W_{}(N+1) = 2 W_{}(N)+W_{}(N)^2 \quad ,
\end{equation}
where the first term reflects the left/right decisions, the second the splittings.
We have $W_{}(N) = 2^{(2^{N-1})}-1$, and find
that $l=2$, $c_0^{(2)}=1$ and $c_j^{(2)}=0$, for $j \geq 1$, and $d_0 =0$, $d_1=0$ and $d_2 \equiv d^{(2)}_0 = 1$.
The corresponding extensive entropy is $g(x) \sim x \, \log \log(x)$ and scales as $S(W) \sim \log \log W$.
Because the coefficients are not negative, SK axioms are applicable.
However, even though all correction scaling exponents are zero, the system cannot be described in terms of $(c,d)$-entropies,
because $l=2$.
We would naively obtain that $c=1$ and $d=0$, which would wrongly correspond to Tsallis entropy. 
Alternatively, we can think of an example of a spin system with the same scaling exponents. In this case, $N$ would not describe the size of a system, but its dimension. For $N=1$, we would have two particles on the line, for $N=2$ we have 4 particles forming a square,
for $N=3$ we have a cube with 8 particles in its vertices, etc. In general, we can think of a spin system of particles sitting
on the vertices of a $N$-dimensional hypercube.
The number of particles is naturally $2^N$ and for two possible spins we obtain $W(N) = 2^{(2^N)}$.

\section{Conclusions}
We introduced a comprehensive classification of complex systems in the thermodynamic limit based on the rescaling properties of their phasespace volume.
From a scaling-expansion of the phasespace growth with system size, we obtain a set of scaling exponents,
which uniquely characterize the statistical structure of the given system.
Restrictions on the scaling exponents can be obtained with further information about the system.
In this context we discuss the first three Shannon-Khinchin axioms, which are valid for many complex systems. The set of exponents further determine the scaling exponents of the corresponding extensive entropy,
which plays a central role in the thermodynamics of statistical systems. Thermodynamics is not the only context where entropy appears.
As was shown in \cite{thurner_corominas_hanel17} for many complex systems the functional expressions for entropy depend on the context,
in particular if one talks about the thermodynamic (extensive) entropy, the information theoretic entropy,
or the entropy that appears in the maximum entropy principle.
It remains to be seen if for super-exponential systems there exists an underlying relation between the scaling exponents of the extensive entropy,
and the exponents obtained from a information theoretic, or maximum entropy description of the same complex systems.

\section*{Acknowledgements}
We thank the participants of CSH workshop for helpful initial discussion, in particular Henrik Jeldtoft Jensen, Tam\'{a}s S\'{a}ndor Bir\'{o}, Piergiulio Tempesta and Jan Naudts. This work was supported by the Austrian Science Fund (FWF) under project I3073.



\section*{References}
\bibliographystyle{iopart-num}
\bibliography{references}

\providecommand{\newblock}{}
\begin{thebibliography}{10}
\expandafter\ifx\csname url\endcsname\relax
  \def\url#1{{\tt #1}}\fi
\expandafter\ifx\csname urlprefix\endcsname\relax\def\urlprefix{URL }\fi
\providecommand{\eprint}[2][]{\url{#2}}

\bibitem{Beckenstein1974}
Bekenstein J~D 1974 {\em Phys. Rev. D\/} {\bf 9}(12) 3292--3300

\bibitem{Hawking1974}
Hawking S~W 1974 {\em Nature\/} {\bf 248} 30--31

\bibitem{Thirring1970}
Thirring W 1970 {\em Z. Phys. A\/} {\bf 235} 339--352

\bibitem{biro}
Bir\'{o} T~S, Czinner V~G, Iguchi H and V\'{a}n P 2018 {\em Phys. Lett. B\/}
  {\bf 782} 228 -- 231

\bibitem{thurner_corominas_hanel17}
Thurner S, Corominas-Murtra B and Hanel R 2017 {\em Phys. Rev. E\/} {\bf 96}(3)
  032124

\bibitem{Sato2005}
Tsallis C, Gell-Mann M and Sato Y 2005 {\em PNAS\/} {\bf 102} 15377--15382

\bibitem{hanel-thurner11a}
Hanel R and Thurner S 2011 {\em Europhys. Lett.\/} {\bf 93} 20006

\bibitem{kaniadakis02}
Kaniadakis G 2002 {\em Phys. Rev. E\/} {\bf 66}(5) 056125

\bibitem{Anteneodo1999}
Anteneodo C and Plastino A~R {\em J. Phys. A\/} {\bf 32} 1089

\bibitem{hanel-thurner11b}
Hanel R and Thurner S 2011 {\em Europhys. Lett.\/} {\bf 96} 50003

\bibitem{Jensen16}
Jensen H~J, Pazuki R~H, Pruessner G and Tempesta P {\em J. Phys. A\/} {\bf 51}
  375002

\bibitem{Hanel_thurner13}
{Thurner} S and {Hanel} R 2011 {\em ArXiv e-prints\/} (\textit{Preprint}
  \eprint{1104.2070})

\bibitem{copson}
Copson E~T 1965 {\em Asymptotic Expansions\/} Cambridge Tracts in Mathematics
  (Cambridge University Press)

\end{thebibliography}

\appendix
\section{Supplementary material}
\subsection*{Shannon-Khinchin axioms}
The Shannon-Khinchin axioms read:
\begin{itemize}
\item (SK1) Entropy is a continuous function of the probabilities $p_i$ only, and should not explicitly depend on any other parameters.
\item (SK2) Entropy is maximal for the equi-distribution $p_i=1/W$.
\item (SK3) Adding a state $W+1$ to a system with $p_{W+1}=0$ does not change the entropy of the system.
\item (SK4) Entropy of a system composed of 2 sub-systems $A$ and $B$, is $S(A+B)=S(A)+S(B |A)$.	
\end{itemize}
They state requirements that must be fulfilled by any entropy. For ergodic systems all four axioms hold. For
non-ergodic ones the composition axiom (SK4) is explicitly violated, and only the first three (SK1-SK3) hold. If all four axioms hold
the entropy is uniquely determined to be Shannon's; if only the first three axioms hold, the entropy is given by the $(c,d)$-entropy
\cite{hanel-thurner11a,hanel-thurner11b}. The SK axioms were formulated in the context of information theory but are also sensible for
many physical and complex systems.

Given a trace form of the entropy as in Eq. (\ref{eq:trace}), the SK axioms imply the restrictions on $g(x)$:
{(SK1)} implies that $g$ is a continuous function, {(SK2)} means that $g(x)$ is concave,
and {(SK3)} that $g(0)=0$. For details, see \cite{hanel-thurner11a}.

\subsection*{Rescaling in the thermodynamic limit}
We first prove a theorem which determines the general form of rescaling relations in the thermodynamic limit for any general function.
\begin{theorem}
 Let $g(x)$ be a positive, continuous function on $\mathds{R}^+$. Let us define the function $z(\lambda): \mathds{R}^+ \rightarrow \mathds{R}^+$
\begin{equation}
z(\lambda) := \lim_{x \rightarrow \infty} \frac{g(r^{(n)}_\lambda(x))}{g(x)}.
\end{equation}
Then, $z(\lambda) = \lambda^c$ for some $c \in \mathds{R}$.
\end{theorem}
\begin{proof}
From the definition of $z(\lambda)$, it is straightforward to show that $z(\lambda \lambda') = z(\lambda) z(\lambda')$, because
\begin{eqnarray*}
	z(\lambda \lambda') 	=  \lim_{x \rightarrow \infty} \frac{g(r^{(n)}_{\lambda \lambda'}(x))}{g(x)}
					= \lim_{x \rightarrow \infty} \frac{g(r^{(n)}_{\lambda \lambda'}(x))}{g(r^{(n)}_{\lambda}(x))} \frac{g(r^{(n)}_{\lambda}(x))}{g(x)}\\
					= \lim_{r^{(n)}_\lambda(x) \rightarrow \infty }
						 \frac{g(r^{(n)}_{\lambda'}[r^{(n)}_{\lambda}(x)])}{g(r^{(n)}_{\lambda}(x))}  \lim_{x \rightarrow \infty} \frac{g(r^{(n)}_\lambda(x))}{g(x)}
					= z(\lambda') z(\lambda) \quad .
\end{eqnarray*}
For the computation we used the group property of rescaling  in Eq. (\ref{eq:resc}) and the continuity of $g$.
The only class of functions satisfying the functional equation above are power functions, $z(\lambda) = \lambda^c$.
\end{proof}
Let us take the first scaling relation of the sample space $W(r^{(0)}_\lambda(N)) = W(\lambda N)$.
From the previous theorem we obtain
\begin{equation}
	\frac{W(\lambda N)}{W(N)} \sim \lambda^{c_{0}}  \ \Rightarrow \ W(r^{(0)}_\lambda(N)) \sim r^{(0)}_{\lambda^{c_{0}}}(W(N)) \quad .
\end{equation}
It may happen that $c_{0}$ is infinite.
Thus, we may need to use higher-order scaling for the sample space, i.e.,
$r^{(l)}_{\lambda^{c_{0}}}(W(N))$, as shown in the main text.
$l$ is determined by the condition that the scaling exponent should be finite.
The first correction term is given by the scaling $W(r^{(1)}_\lambda(N)) = W(N^\lambda)$.
To obtain the sub-leading correction, we have to factor out the leading growth term.
This means that the scaling relation for the first sub-leading correction looks like
\begin{equation}
	\frac{(\log^{(l)} W(N^\lambda))/N^{c^{(l)}_{0} \lambda}}{(\log^{(l)} W(N))/N^{c^{(l)}_{0}}} \sim \lambda^{c^{(l)}_{1}} \quad ,
\end{equation}
which is again a consequence of the above theorem.
To obtain the corresponding scaling relations for higher-order scaling exponents for the sample space (\ref{eq:sub}),
we need to factor out all previous terms corresponding to lower-order scalings, so the scaling relation looks like
\begin{equation}\label{eq:sub}
	\frac{\log^{(l)} W(r^{(k)}_\lambda(N))}{\log^{(l)} W(N)} \prod_{j=0}^{k-1}
	\left(\frac{\log^{(j)}(r^{(k)}_\lambda(N))}{\log^{(j)}(N)}\right)^{- c^{(l)}_{j}} \sim \lambda^{c^{(l)}_{k}}
\end{equation}
Because the left-hand side of this relation has the form of the function $z$ appearing in the theorem,
the validity of the relation is satisfied for $N \rightarrow \infty$.
Similarly, we can deduce the relations for scaling exponents that are associated with the extensive entropy.

\subsection*{Asymptotic expansion in terms of nested logarithms}
The asymptotic representation of $W(N)$ is obtained by the rescaling that corresponds to the Poincar\'{e} asymptotic expansion \cite{copson} of
$\log^{(l+1)}(W)$ in terms of $\phi_n(N) = \log^{(n+1)}(N)$ for $N \rightarrow \infty$.
Let us consider a function $f(x)$ with a singular point at $x_0$. It is possible to express its asymptotic properties in the neighborhood of $x_0$ in terms of the asymptotic series of functions $\phi_n(x)$, if $f(x) = {\cal O} (\phi_0(x))$ and $\phi_{n+1}(x) =  {\cal O} (\phi_n(x))$.
The series is given as
\begin{equation}
f(x) = \sum_{j=0}^{k} c_j \phi_j(x) + {\cal O} (\phi_j(x)) \quad .
\end{equation}
The coefficients can be calculated from the formulas in \cite{copson}
\begin{equation}
	c_k = \lim_{x \rightarrow x_0} \frac{f(x) - \sum_{j=0}^{k-1}c_j \phi_j(x)}{\phi_k(x)} \quad .
\end{equation}
In our case, i.e., for $N \rightarrow \infty$ and $\phi_n(N) = \log^{(n+1)}(N)$   the function $\log^{(l+1)}(W)$ can be expressed (for appropriate $l$) in terms of this series, and the coefficients $c_k^{(l)}$ are given by
\begin{eqnarray}
	c_k^{(l)} 	&=& \lim_{N \rightarrow \infty} \frac{\log^{(l+1)}(W) - \sum_{j=0}^{k-1}c_j^{(l)} \log^{(j+1)}(N)}{\log^{(k+1)}(N)}\nonumber\\
 			&=& \lim_{N \rightarrow \infty} \frac{\log\left(\log^{(l)}(W)/\prod_{j=0}^{k-1}\log^{(j)}(N)^{c_j^{(l)}}\right)}{\log^{(k+1)}(N)}\nonumber \quad .
\end{eqnarray}
Using L'Hospital's  rule and the derivative of the nested logarithm
\begin{equation}
	\frac{\mathrm{d} \log^{(n)}(x)}{\mathrm{d} x} = \frac{1}{\prod_{j=0}^{n-1} \log^{(j)}(x)}  \quad ,
\end{equation}
a straightforward calculation yields Eq. (\ref{eq:clk}).

\subsection*{Derivation of $g^{{(l,n)}}_{\left(d^{(l)}_0,\dots,d^{(l)}_n\right)}$}
Which entropy functional that fulfills axioms ({SK} 1-3)?
The choice is not unique, but a concrete entropy functional serves as a representative of the class in the thermodynamic limit.
The requirements imposed by the first three {SK} axioms are: $g(x)$ is continuous,
 $g(x)$ is concave, and $g(0) = 0$. From  Eq. (\ref{eq:g}) we have,
$g(x) \sim x \prod_{j=0}^{n} [\log^{(j+l)}\left(\frac{1}{x}\right)]^{d^{(l)}_j}$ for $x \rightarrow 0$,
which gives us the scaling for the values around zero.
Unfortunately, the presented form cannot be extended to the full interval $[0,1]$, because the domain of $\log^{(n)}(1/x)$ is $(0,1/\exp^{(n-2)}(1))$.
This can be fixed by replacing $\log^{(n)}$ by $[1+\log]^{(n)} = 1+ \log(1+ \log(\dots))$,
which is defined on the whole domain $(0,1]$, where $\lim_{x \rightarrow 0} [1+\log]^{(n)}(1/x) = +\infty$ and $[1+\log]^{(n)}(1) = 1$.
The scaling remains unchanged for $x \rightarrow 0$.

The second problem is that in general the function is not concave. For this we introduce the transformation
\begin{equation}
	f^\star(x) = \int_0^x \frac{f(y)}{y} \mathrm{d} y\, .
\end{equation}
The original function can be obtained by
\begin{equation}
	f(x) = x \frac{\mathrm{d} f^\star(x)}{\mathrm{d} x}.
\end{equation}
This transform turns an increasing/decreasing function to a convex/concave function, while the scaling for $x \rightarrow 0$ remains unchanged.
Let us write the function $g$ in the form of the transform
\begin{equation}
	g(x) \sim \int_0^{x} \prod_{j=0}^{n} \left[ [1+\log^{(n)}]\left(\frac{1}{y}\right)\right]^{d_j} \, \mathrm{d} y \quad .
 \end{equation}
Axiom ({SK3}) means $g(0) = 0$. This requires that the integrand should not diverge faster than $1/x$ for $x \rightarrow 0$.
This can be fulfilled for $d_0 \equiv d_0^{(0)} < 1$.

Because $[1+\log]^{(n)}(1/x)$ is a decreasing function, $g(x)$ is automatically concave if $d_n \geq 0$, since a product of positive,
decreasing functions is also decreasing.
However, for $d_n < 0$, $[1+\log]^{(n)}(1/x)^{d_n}$ is an increasing function from zero to one and the whole product may not be decreasing.
In order to solve this issue, we introduce a set of constants $a_i$ and write $g(x)$ in the form
\begin{equation}
	g(x) = \int_0^{x} \prod_{j=0}^{n} \left[a_j + [1+\log]^{(n)}\left(\frac{1}{y}\right)\right]^{d_j} \, \mathrm{d} y \quad .
 \end{equation}
The constants $a_i$ can be chosen to ensure that the integrand is a decreasing function.
We assume $a_i \geq -1$ to avoid problems with powers of negative numbers.
The second derivative of $g(x)$, i.e., the first derivative of the integrand is an increasing function and
$\frac{\mathrm{d}^2 g(x)}{\mathrm{d} x^2}|_{x \rightarrow 0^+} = - \infty$ for $d_l > 0$.
For $d_l < 0$, the entropy cannot be concave, so $d_l > 0$ is the restriction given by ({SK2}).
To obtain a negative second derivative on the whole domain $[0,1]$, it is therefore enough to investigate
$\frac{\mathrm{d}^2 g(x)}{\mathrm{d} x^2}|_{x = 1}$, which leads to the condition
\begin{equation}\label{eq:2d}
 	\left(\prod_{j=l}^{n} (1+a_j)^{d_j-1} \right) \left(- \sum_{j=l}^{n} \frac{d_j}{1+a_j}\right) \leq 0 \quad .
\end{equation}
Because $d_0^{(l)} \equiv d_l > 0$, we can choose $a_l = 0$. In the following terms, i.e., for $i > l$, $d_i$ can be both positive and negative.
Positive $d_i$ pose no problem, because the term corresponding to $d_i$, i.e. $-d_i/(1+a_i)$ is negative, so we can choose $a_i=0$.
When all $d_i$ are negative we can compensate the positive contribution of the negative terms by diminishing them through choice of appropriate $a_i$.
If we choose
\begin{equation}
	1+a_i = - \frac{d^{(l)}_i}{n d^{(l)}_0} \quad ,
\end{equation}
then Eq. (\ref{eq:2d}) becomes zero. If this is given together with previous results and summarize it as
\begin{equation}
	a_i  =  \max\left\{ -1 - \frac{d^{(l)}_i}{n d^{l}_0},0\right\} \quad ,
\end{equation}
which has been presented as in Eq. (\ref{eq:ai}) in the main text.
Clearly, this is not the only possible choice. Note that for all $d^{(l)}_i > 0$, one may even choose $a_i = - 1$.
On the other hand, for the case of magnetic coin model, one obtains that for $a_0 = 0$, $a_1 = 0$ as well.

Finally, let us show the connection to $(c,d)$-entropy derived in \cite{hanel-thurner11a}.
In this case, we assume only  $d_0$ and $d_1$ can be non-zero, which leads to
\begin{equation}
	g^{{(0,1)}}_{\left(d_0,d_1\right)}(x) = \int_0^{x} (1/y)^{d_0} (1+a_1+\log(1/x))^{d_1} \, \mathrm{d} y \quad .
\end{equation}
By the choice $a_1 = -1+\frac{1}{1-d_0}$, we get
\begin{equation}
	g_{(c,d)}(x) = \frac{e}{c^{d+1}} \Gamma(1+d,1-c \log x) \quad ,
\end{equation}
for $c = 1-d_0$ and $d = d_1$, which is nothing else than the Gamma entropy of \cite{hanel-thurner11a}.

\subsection*{Ordering of processes and classes of equivalence}
The set of scaling exponents form natural classes of equivalence with natural ordering.
Consider two discrete random processes $X(N)$ and $Y(N)$ with sample spaces $W_X(N)$ and $W_Y(N)$, respectively.
The corresponding sets of scaling exponents are denoted by $\mathcal{C}_X =\{c^{(l)}_0,c^{(l)}_1,\dots\}$,
and $\mathcal{C}_Y =\{\tilde{c}^{(\tilde{l})}_0,\tilde{c}^{(\tilde{l})}_1,\dots\}$.
One can introduce an ordering based on the scaling exponents. We write
\begin{equation}
X \prec Y \ (\ \mathcal{C}_X \prec \mathcal{C}_Y) \ \mathrm{if} \ \left\{
                                                                                   \begin{array}{ll}
                                                                                     l < l'\\
                                                                                     l = l', c^{(l)}_0 < \tilde{c}^{(\tilde{l})}_0\\
                                                                                     l = l', c^{(l)}_0 = \tilde{c}^{(\tilde{l})}_0, c^{(l)}_1 < \tilde{c}^{(\tilde{l})}_1\\
\mathrm{etc.}
                                                                                   \end{array}
                                                                                 \right. \quad .
\end{equation}
This is equivalent to lexicographic ordering. One can also introduce an ordering,
which takes into account only certain a number of correcting terms. So, for example
\begin{equation}
X \prec_0 Y \ (\ \mathcal{C}_X \prec_0 \mathcal{C}_Y) \ \mathrm{if} \ \left\{
                                                                                   \begin{array}{ll}
                                                                                     l < l'\\
                                                                                     l = l', c^{(l)}_0 < \tilde{c}^{(\tilde{l})}_0\\
\end{array}
                                                                                 \right. \quad .
\end{equation}
Similarly, one can define $\prec_k$, which takes into account only $k$ correction terms.
Additionally, it is possible to introduce an equivalence relation
\begin{equation}
	X \sim Y \ \mathrm{if}  \ \mathcal{C}_X \equiv \mathcal{C}_Y \ \Rightarrow \  l = l'; \   c^{(l)}_i = \tilde{c}^{(\tilde{l})}_i \ \forall i \quad .
\end{equation}
and also equivalence up to certain correction
\begin{equation}
	X \sim_k Y \ \mathrm{if}  \ \mathcal{C}_X \equiv \mathcal{C}_Y \ \Rightarrow \  l = l'; \   c^{(l)}_i = \tilde{c}^{(\tilde{l})}_i \ \forall i \leq k \quad .
\end{equation}
As an example, for magnetic coin model and random walk we have that $X_{\rm MC} \sim_0 X_{\rm RW}$, but  $X_{\rm MC} \not\sim X_{\rm RW}$.

\subsection*{Construction of a ``representative process''}
To understand the mechanism of how the scaling exponents correspond to the structure of a random process,
let us discuss a simple procedure to generally obtain processes with given scaling exponents $c^{(l)}_k$.
We start with a random variable $X_0$ with $N$ possible outcomes, so that $W_{X_0}(N) = \{1,\dots,N\}$.
The scaling exponents of this process are naturally $c^{(0)}_0 = 1$ and $c^{(0)}_k = 0$ for $k \geq 1$.
Let us construct a new variable by choosing subsets of $W_{X_0}(N)$.

First we can create all possible subsets of $W_{X_0}(N)$.
This defines a new variable $X_1$ with $W_{X_1}(N) = 2^{W_{X_0}(N)}$, and we get $c^{(1)}_0 = 1$.
Generally, the transform
\begin{equation}
	\mathfrak{2}: X \rightarrow \mathfrak{2}^{X} \quad ,
\end{equation}
where $\mathfrak{2}^X$ denotes a variable on all subsets of $X$.
One can easily show that this results in a  shift of scaling exponents $c^{(l)}_k \rightarrow c^{(l+1)}_k$, and $d^{(l)}_k \rightarrow d^{(l+1)}_k$,
because $W_{\mathfrak{2}^{X}}(N) = 2^{W_X(N)}$.
The interpretation of this transformation is the following:
Consider an ordinary random walk with two possible steps. If $X_0(N)$ denotes a number of steps of a random walker,
then $X_1(N) = \mathfrak{2}^{X_0(N)}$ denotes the number of possible paths.
When we apply the transform again, we obtain $X_2(N) = \mathfrak{2}^{X_1(N)}$.
This denotes the number of possible configurations of a random walk cascade, etc.
As a result, by more applications of $\mathfrak{2}$, we obtain processes with more complicated structure of the respective phasespace.

To construct processes with arbitrary exponents, let us think about a procedure, where we create only partial subsets,
which number $p(N)$ can be between
$N$ (no partitioning) and $2^N$ (full partitioning)
We denote this procedure by $\mathfrak{P}$.
This process can be understood as process corresponding to a correlated random walk. This means that not every step of the walk is independent, but some steps can be determined by the previous steps, which diminishes the number of possible configurations when compared to the uncorrelated random walk. The resulting random process is obtained as the composition of $l$ uncorrelated random walks (full partitioning) and a correlated random walk
\begin{equation}
X = \mathfrak{2}^{(l)}[\mathfrak{P}(X_0)] \quad .
\end{equation}
Let us now focus on the construction of correlated random walk with a pre-determined number of states given by $p(N)$.

First we consider the full set of subsets of $N$ elements with natural ordering,
\begin{equation}
	W_{\mathfrak{2}^X} =\left\{\{\},\{1\},\{2\},\dots,\{1,\dots,n\}\right\} \quad .
\end{equation}
The correlations can be represented by merging subsets to $p(N)$ sequences of length $\{s(1),\dots,s(p(N))\}$, i.e.,
\begin{equation}
	W_{\mathfrak{P}(X)} =\left\{ \underbrace{\{ \{\},\{1\},\{2\},\dots\}}_{s(1)},\dots,\underbrace{  \left\{ \dots, \left\{ 1,\dots,n \right\} \right\} }_{s(p(N))}\right\} \quad .
\end{equation}
This means that after one independent step, there are $s(1)-1$ dependent steps, after the second independent step, there are $s(2)-1$ dependent steps, etc. Let us determine the form of function $s$ for given $p(N)$. The function $s$ can be obtained from
\begin{equation}
	\sum_{i=1}^{p(N)} s(i) = 2^N \quad .
\end{equation}
In the limit of large $N$ we can assume that the function $s$ does not depend on $N$, i.e., is a priori given by the scaling exponents of the system. Let us also assume, without loss of generality, that $s$ is an increasing function (we can neglect the last cell, because its size is determined by the size of previous cells).
For $N \gg 1$, we approximate the sum by the integral and obtain
\begin{equation}
	\int_0^{p(N)} s(i) \mathrm{d} i \sim 2^N \quad .
\end{equation}
Denoting $S(m) = \int_0^m s(y) \mathrm{d} y$, and substituting $x=p(N)$, we recast the previous equation as,
$S(x) = 2^{p^{-1}(x)}$,
where $p^{-1}$ denotes the inverse function of $p$. The function $s(x)$ can be therefore determined as
\begin{equation}
	s(x) = \frac{\mathrm{d} \,  2^{p^{-1}(x)}}{\mathrm{d} x} = \frac{2^{p^{-1}(x)}}{p'(p^{-1}(x))}\quad .
\end{equation}
Some examples for $s(x)$ for a corresponding $p(N)$ are
\begin{itemize}
\item $p(N) = 2^N$, i.e., full partitioning corresponding to uncorrelated random walk. In this case, we obtain that $s(x) = const.$, as expected.
\item $p(N) = N$, i.e., no partitioning to maximally correlated random walk. We obtain that $s(x) \sim 2^x$, which can be seen from the relation $\sum_{i}^{N} 2^i \sim 2^N$.
\item $p(N) = N \log N$, which corresponds to the correction in the magnetic coin model.
	In this case, $s(x) \sim 2^{W(x)}/\log(W(x))$, where $W(x)$ is the Lambert W-function.
\end{itemize}

\end{document}